\documentclass[11pt, reprint, notitlepage,tightenlines,nofootinbib,superscriptaddress,aps,prl]{revtex4-2}
\usepackage[cal=boondox]{mathalfa}        
\usepackage{amsmath,amssymb,amsthm,microtype,mathtools}
\usepackage{enumitem}
\usepackage{graphicx}
\usepackage{float}
\usepackage[section]{placeins}
\usepackage{xcolor}
\usepackage{tikz}
\usepackage{pgfplots}
\usetikzlibrary{arrows.meta,calc,decorations.pathreplacing,positioning,fit,backgrounds,shapes.geometric,patterns,shadows,shapes.misc}
\usepackage{comment}
\usepackage{mathpazo}
\usepackage{braket}
\usepackage{booktabs}

\definecolor{mitred}{RGB}{163, 31, 52} 
\usepackage[colorlinks=true,linkcolor=mitred,citecolor=mitred,urlcolor=mitred]{hyperref}

\pgfplotsset{compat=1.18}
\pgfmathdeclarefunction{exactpt}{1}{%
  \pgfmathparse{0.5*(1 + floor(1/(#1))*(#1)^2 + (1-floor(1/(#1))*(#1))^2)}%
}

\newtheorem{theorem}{Theorem}
\newtheorem{lemma}[theorem]{Lemma}
\newtheorem{proposition}[theorem]{Proposition}

\DeclareMathOperator{\Id}{\mathbb{I}}

\newcommand{\abs}[1]{\left\lvert #1\right\rvert}

\newcommand{\tr}[1]{\textnormal{tr}\left[#1\right]}
\newcommand{\ketbra}[2]{\ket{#1}\!\!\bra{#2}}

\newcommand{\eps}{\varepsilon}

\newcommand{\BC}{\mathcal{B}}

\newcommand{\DC}{\mathcal{D}}
\newcommand{\EC}{\mathcal{E}}

\newcommand{\HC}{\mathcal{H}}

\newcommand{\PC}{\mathcal{P}}

\newcommand{\SC}{\mathcal{S}}

\newcommand{\UC}{\mathcal{U}}

\newcommand{\lsec}[1]{\textit{#1.---}}

\hypersetup{
  pdftitle={Tight Bounds for Purity and Product Testing from Partial Transposition},
  pdfauthor={Oren Akresh, Jacob Beckey},
  pdfsubject={Quantum property testing},
  pdfkeywords={product test, quantum property testing, multipartite entanglement, PPT, restricted measurements}
}

\begin{document}
\title{Tight Bounds for Purity and Product Testing from Partial Transposition}
\author{Oren Akresh}
\affiliation{Central High School, Champaign, IL 61820, USA.}

\author{Jacob Beckey}
\email{jbeckey@illinois.edu}
\affiliation{Department of Mathematics, University of Illinois at Urbana-Champaign, Urbana, IL 61801, USA.}
\affiliation{IQUIST, University of Illinois Urbana-Champaign, Urbana, Illinois 61801, USA.}
\affiliation{Department of Physics, University of Rhode Island, Kingston, RI 02881, USA}

\date{\today}

\begin{abstract}

Coherent measurements across multiple copies of an unknown quantum state can substantially reduce the number of samples required to learn its properties, but remain experimentally challenging. Current experiments typically prepare and measure one copy at a time, potentially adapting later measurements to earlier outcomes. A central challenge is adaptivity, which makes the space of possible measurement strategies difficult to characterize. The positive-partial-transpose (PPT) relaxation bypasses this complexity by considering a larger, mathematically tractable class of measurements, at the risk of weakening the resulting bounds. Here we show, surprisingly, that the relaxation loses nothing at the level of asymptotic sample complexity for two fundamental tasks: purity testing and product testing. In both cases, lower bounds against the full class of PPT measurements are matched by  nonadaptive single-copy protocols. Moreover, our proof requires only basic symmetric subspace identities, providing a simple route to sharp lower bounds for adaptive single-copy measurements.
\end{abstract}

\maketitle

\lsec{Introduction} Coherent measurements across multiple copies of an unknown quantum state can drastically reduce the total number of copies needed to learn properties of that state~\cite{odonnell2016Efficient,haah2017SampleOptimal,chen2022Exponential,huang2022Quantum,ye2025Exponential,noller2025infinite,bubeck2020Entanglement,chen2024optimal}. Yet even high-fidelity two-copy measurements remain at the experimental frontier~\cite{bluvstein2022quantum,bluvstein2023Logical,daguerre2025Experimental,miller2026Experimental}, while coherent measurements across many copies, as required by numerous optimal algorithms, are squarely out of reach. State characterization therefore typically proceeds one copy at a time, with later measurements potentially adapted to all previous outcomes. This ease of implementation comes at a steep cost in terms of the number of samples needed to characterize the unknown quantum state~\cite{chen2022Exponential,huang2022Quantum,ye2025Exponential,noller2025infinite}. 

A particularly stark separation between single- and multi-copy measurements arises for the ubiquitous task of purity estimation. Given coherent access to two identical copies of a $d$-dimensional quantum state, one can estimate purity using a SWAP test~\cite{barenco1997Stabilization,ekert2002Direct} with a number of repetitions independent of $d$. In contrast, even a potentially adaptive single-copy measurement protocol requires $\Omega(\sqrt{d})$ samples to estimate the purity to constant precision~\cite{chen2022Exponential,gong2024sample}. For a $k$-qubit state, $d=2^k$, so this represents an \textit{exponential separation} between the sample complexity of one- and two-copy protocols.

A sample complexity lower bound implies that \textit{there cannot exist} an algorithm using fewer samples and proving such results can be highly non-trivial, especially when one allows the learner to choose measurements adaptively based on all prior outcomes. Many recent tight sample complexity lower bounds have been proved in this setting using the so-called \textit{learning tree formalism}, inspired by tools from theoretical computer science~\cite{chen2022Exponential}. This technique involves the determination of probability distributions over measurement transcripts and uses results from classical statistics to prove hardness of distinguishing these resultant distributions~\cite{chen2022Exponential,chen2022When,arunachalam2026Optimal,gong2024sample,beckey2025Product,chen2024Optimala}. 

A complementary approach, inspired by the entanglement theory~\cite{peres1996Separability,horodecki1996Separability,rains2001semidefinite} and state discrimination~\cite{eggeling2002Hiding,matthews2009Distinguishability} literature, is taken in Refs.~\cite{harrow2023Approximate,hinsche2025SingleCopy}. Here the authors relax the set of allowable measurements to a mathematically tractable set classified by the fact that the measurement operators remain positive after partial transposition (PPT). Then, one uses the approximate orthogonality of permutation operators~\cite{harrow2023Approximate,hinsche2025SingleCopy} to prove the desired sample complexity lower bounds. Unfortunately, as we discuss below, this technique typically yields weaker sample lower bounds than the learning tree methods for the same tasks.

The PPT relaxation is mathematically tractable but can be loose, sometimes exhibiting strict separations from protocols based on local operations and classical communication (LOCC) in state discrimination and entanglement manipulation~\cite{beigi2010Approximating,bennett1999Quantum,cheng2023Discrimination,liu2023Complexity}. Despite enlarging the set of adaptive single-copy measurements, our lower bounds are asymptotically tight, being saturated by previously known nonadaptive single-copy measurements~\cite{chen2022Exponential,gong2024sample,beckey2025Product}. To our knowledge, these are the first quantum property testing problems for which relaxing adaptive single-copy measurements to PPT measurements preserves the optimal dimension dependence of the sample complexity.

Remarkably, our proof requires only basic facts from linear algebra and the church of the symmetric subspace~\cite{harrow2013Church}. Moreover, it does not require the dimension of the state $d$ to be larger than the number of copies $n$ (i.e. the stable regime of Schur-Weyl) as in~\cite{harrow2023Approximate}. Our partial-transposition approach thus avoids the dimensional restrictions of approximate orthogonality and complements the learning-tree formalism, providing an alternative route to sample-complexity lower bounds when tree-based analyses become technically difficult, using tools familiar across quantum physics.

\lsec{Property Testing the PPT Relaxation} In quantum property testing, one must determine, with success probability at least $2/3$, whether an unknown state $\rho$ satisfies a property $\PC\subseteq\BC(\mathbb{C}^d)$ or is $\eps$-far in trace distance from every state in $\PC$, given the promise that one of these holds~\cite{montanaro2016Survey}. To prove a sample-complexity lower bound, one typically constructs distributions $\EC_0$ and $\EC_1$ supported on states satisfying the respective promises. Averaging over the unknown state drawn from $\EC_b$ gives the $n$-copy state
\begin{align}
\omega_b^{(n)}
\coloneq
\underset{\rho\sim\EC_b}{\mathbb{E}}
\left[\rho^{\otimes n}\right],
\qquad b\in\{0,1\}.
\end{align}
A two-outcome measurement $\{M,\Id-M\}$ on $(\mathbb{C}^d)^{\otimes n}$ distinguishes these ensembles with bias
$\abs{\tr{M(\omega_0^{(n)}-\omega_1^{(n)})}}$. Any tester that succeeds on every promised state with probability at least $2/3$ must therefore achieve bias at least $1/3$ on these averaged states. Thus, an upper bound $B(n,d)$ on the bias over the relevant measurement class implies that constant-success testing requires $B(n,d)=\Omega(1)$. As we will see shortly, this allows one to infer a sample complexity lower bound from the upper bound on the distinguishing bias.

Treating each copy as a separate party, any adaptive single-copy protocol can be viewed as one-way LOCC measurement $\{M,\Id-M\}$ on $\HC^{\otimes n}$ across the partition $\HC_1:\cdots:\HC_n$. Following Harrow~\cite{harrow2023Approximate}, we bypass the complexity of adaptivity by relaxing through the strict inclusions~\cite{harrow_testing_2010,chitambar2014Everything}
\begin{align}
\mathrm{LOCC}\subset\mathrm{SEP\text{-}BOTH}
\subset\mathrm{PPT\text{-}BOTH}.
\end{align}
Here, BOTH requires both measurement effects to belong to the indicated class; in particular, PPT-BOTH is equivalent to
\begin{align}\label{eq:PPT-BOTH}
0\preceq M^{\Gamma_S}\preceq\Id
\qquad\forall,S\subseteq[n],
\end{align}
where $\Gamma_S$ denotes partial transposition of the registers in $S$. This relaxation is substantial: unlike separable measurements, PPT-BOTH permits effects that are bound entangled~\cite{horodecki1997Separability,bennett1999Unextendible,cheng2023Discrimination}. Harrow combines this relaxation with his results regarding the approximate orthogonality of permutation operators to show that any PPT-BOTH measurement used to distinguish $n$ copies of a Haar random pure state from $n$ copies of the maximally mixed state, will achieve a bias $b \leq O(n^2/\sqrt{d})$. The assumption that a
PPT-BOTH purity tester exists implies $1/3 \leq b$. This lower bound is violated unless one takes $n=\Omega(d^{1/4})$, yielding a sample complexity lower bounds that is quadratically weaker than the $\Omega(\sqrt{d})$ bound obtained by the learning tree formalism~\cite{chen2022Exponential,gong2024sample}.

This is not unique to purity testing. The PPT-BOTH relaxation combined with approximate orthogonality tools also yield weaker bounds for the two-party case of product testing, where the same approach gives only $\Omega(d^{1/8})$, again below the optimal $\Omega(\sqrt d)$ scaling~\cite{liu2025Separation,beckey2025Product}. We show that, at the level of asymptotic sample complexity, these gaps arise from the analysis rather than the relaxation: every PPT-BOTH tester for either task requires $\Omega(\sqrt d)$ samples. We begin with purity testing, which we emphasize is a weaker problem than the more experimentally motivated task of purity estimation, thus sample lower bounds on purity testers imply sample lower bounds on purity estimation algorithms.

\lsec{Purity Testing}
Let us first consider testing whether an unknown state is pure or $\eps$-far from all pure states. Let $\psi\coloneq\ketbra{\psi}{\psi}$ and define the $n$-fold Haar averaged state as
\begin{align}
\sigma^{(n)}
\coloneq
\underset{\ket{\psi}\sim\mathrm{Haar}}{\mathbb{E}}
\left[\psi^{\otimes n}\right]
=
\frac{\Pi_{\mathrm{sym}}^{(n)}}{\binom{d+n-1}{n}},
\end{align}
where $\Pi_{\mathrm{sym}}^{(n)}$ projects onto the symmetric subspace of $(\mathbb{C}^d)^{\otimes n}$. Distinguishing $n$ copies of a Haar random pure state from $n$ copies of the maximally mixed state $\tau_d \coloneq \Id/d$ reduces to purity testing because every pure state satisfies $d_{\mathrm{tr}}(\psi,\tau_d)
=
1-\frac{1}{d}$. Thus, any valid purity tester could distinguish these ensembles for any constant $\eps \in (0,1-1/d]$. We may now state our first result.

\begin{theorem}[PPT-BOTH purity testing]\label{thm:eps-ind-PPT-purity-LB}
For every $d,n\geq2$ and every PPT-BOTH measurement $\{M,\Id-M\}$,
\begin{align}
\abs{
\tr{
M\left(
\sigma^{(n)}-\tau_d^{\otimes n}
\right)
}
}
\leq
\frac{n(n-1)}{2d}.
\end{align}
Consequently, constant-bias purity testing under PPT-BOTH measurements requires $n=\Omega(\sqrt d)$ samples.
\end{theorem}
While the complete proof is provided in the End Matter, the core technique is sufficiently simple to sketch here. Rather than exploiting symmetry to restrict the form of $M$, as in Ref.~\cite{harrow2023Approximate}, we begin with a general telescoping decomposition that successively depolarizes the last register. Writing $\Delta\coloneq\sigma^{(n)}-\tau_d^{\otimes n}$, we have
\begin{align}
\Delta
=
\sum_{m=2}^n
\left(
\sigma^{(m)}
-
\sigma^{(m-1)}\otimes\tau_d
\right)
\otimes\tau_d^{\otimes(n-m)}.
\end{align}
We then use two elementary facts about inner products of operators on $(\mathbb{C}^d)^{\otimes n}$. First, $\tr{XY}= \tr{X^{\Gamma_S}Y^{\Gamma_S}}$ for all $S \subseteq [n]$. Second, if $X$ is an effect (i.e. $0 \preceq X \preceq \Id$) and $Y$ is traceless and Hermitian, then $\tr{XY} \leq \frac{1}{2}\|Y\|_1$. Thus, inserting our telescoping sum into the expression for the bias and recalling that the PPT both condition in Eq.~\eqref{eq:PPT-BOTH} ensures $M^{\Gamma_S}$ is an effect, we obtain
\begin{align}
    \abs{\tr{M \Delta}}  \leq \frac{1}{2}\sum_{m=2}^n \|(\sigma^{(m)})^{\Gamma_m} - \sigma^{(m-1)}\otimes \tau_d\|_1.
\end{align}
The whole problem then reduces to upper-bounding the trace norm of this operator that has its last tensor factor transposed. A bound with the required $O(m/d)$ scaling is a simple consequence of permutation symmetry.

To see this, recall that $\sigma^{(m)}$ is the normalized projector onto the symmetric subspace and hence a uniform average of permutation operators (cf. Ref~\cite{harrow2013Church}). These permutations can be organized according to what happens to the $m$-th site. For all permutations fixing this last site, the partial transposition acts trivially and the resultant difference of operators yields a correction whose trace norm is at most $O(m/d)$. The remaining permutations can be viewed as a permutation of the first $m-1$ terms followed by a SWAP between the $m$-th site and one of the other $m-1$ sites. Partially transposing the SWAP operator yields the well-known formula
\begin{align}
    P_d((im))^{\Gamma_m} = d \Phi_{i,m}.
\end{align}
where $\Phi_{i,m}$ is the projector onto the maximally entangled state between the $m$-th site and the $i$-th, where $i \in [m-1]$. This simple argument allows us to decompose the operator whose trace norm we wish to bound into the difference of two positive operators with traces equal to $(m-1)/(d+m-1)$. Applying the triangle inequality yields the desired bound
\begin{align}
    \abs{\tr{M \Delta}} \leq \sum_{m=2}^n \frac{m-1}{d} = \frac{n(n-1)}{2d}.
\end{align}

The same argument extends to nonconstant accuracy by considering the Haar orbit of
$\rho_\delta\coloneq(1-\delta)\ketbra{0}{0}+\delta\tau_d$, where $\delta=\Theta(\eps)$ is chosen so that its distance from the set of pure states, $\delta(1-1/d)$, equals $\eps$. The calculation is given in the Supplemental Material.

\begin{proposition}[$\eps$-dependent purity-testing lower bound]
\label{prop:eps-dep-purity-LB}
For $0<\eps\leq1-1/d$, any PPT-BOTH purity tester with distance parameter $\eps$ requires
$n=\Omega(\sqrt{d/\eps})$ samples.
\end{proposition}
Together with the universal Holevo--Helstrom lower bound $\Omega(1/\eps^2)$ for additive purity estimation, this matches the single-copy lower bound of Ref.~\cite[Theorem~6]{gong2024sample}, which uses the learning tree formalism and advanced martingale techniques (cf. Appendix A of Ref.~\cite{chen2024Optimala}).

\lsec{Product Testing} We now turn to the problem of product testing, for which our proof technique also enables a significant strengthening of the best PPT-BOTH result~\cite{harrow2023Approximate} and matches the tight lower bound proved via the tree formalism~\cite{liu2025Separation,beckey2025Product}. Compared to both of the known proofs, our results are again, substantially more elementary.

Suppose we are given a $k$-partite pure state $\ket{\psi} \in (\mathbb{C}^d)^{\tilde{\otimes} k}$, where we have introduced the notation $\tilde{\otimes}$ to distinguish the tensor structures \textit{within} a single copy of $\ket{\psi}$ and \textit{between} multiple copies. Product testing asks one to determine whether the given quantum state is of the form $\ket{\psi} = \ket{\psi_1} \otimes \dotsm \otimes \ket{\psi_k}$ or whether it is far from any such state.

Entangled states are often primitives for quantum information processing tasks, thus an enormous body of literature is dedicated to entanglement characterization~\cite{horodecki2009Quantum,guhne2009Entanglement}. There is excellent motivation, however, for studying the complementary problem of certifying that the given state is actually product. 

The product state ansätze underpin mean-field methods used throughout quantum physics and chemistry, motivating the certification of proximity to the set of product states~\cite{brandao2016ProductState,bakshi2024Learning}. The strongest motivation to study product testing comes from Harrow and Montanaro's seminal work~\cite{harrow_testing_2010} in which they utilize a two-copy measurement based on parallel SWAP tests~\cite{brennen2003observable,mintert2005Concurrence} to construct a product tester with constant sample complexity, which they used to prove the landmark equality $\mathrm{QMA}(k)=\mathrm{QMA}(2)$ for every $k\geq2$, where QMA is the quantum analogue of NP~\cite{jeronimo2026QMA2}. 

Beyond their experimental relevance, the limitations of single-copy product testers have striking complexity-theoretic consequences. As Harrow observes~\cite{harrow2023Approximate}, an efficient LOCC product tester would imply the unexpected collapse $\mathrm{QMA}=\mathrm{QMA}(2)$, resolving a central open problem in quantum complexity theory; with sufficient accuracy, such a tester could also contradict the Unique Games Conjecture~\cite{barak2012Hypercontractivity} or the Exponential Time Hypothesis~\cite{harrow_testing_2010}. Lower bounds on local product testers therefore bear directly on central questions in quantum and classical complexity theory.

In Appendix~D of Ref.~\cite{harrow_testing_2010}, the authors prove the non-existence of an efficient LOCC product test for the special case $n=2$ and $k=2$ using the PPT relaxation and a direct calculation based on the imposed constraints. Attempting to directly prove such a lower bound for $n>2$ quickly becomes intractable, and a forward citation was given to Ref.~\cite{harrow2023Approximate}, where the more advanced approximate orthogonality of permutation operators can be used to show that any PPT-BOTH product test for $k=2$ and $n$ arbitrary must use $n = \Omega(d^{1/8})$ to achieve constant bias. Our second main result significantly strengthens this result and, in light of Ref.~\cite{beckey2025Product}, is tight and saturated by a nonadaptive single-copy measurement.

\begin{theorem}\label{thm:product-testing-LB}
    If $\{M,\mathbb{I}-M\}$ is a PPT-BOTH product tester for $k=2$ acting on $n$ copies of state, then its bias is at most $O(n^2/d)$. Thus, a product tester using only PPT-BOTH measurements requires $n=\Omega(\sqrt{d})$ samples.
\end{theorem}

The proof of this result is closely related to the proof above, but with a different distinguishing task. Consider trying to distinguish $n$ copies of a Haar random pure state on $\mathbb{C}^d \tilde{\otimes} \mathbb{C}^d$ and a Haar random product state on this same space. Given $n$ identical copies, our task is to distinguish the density matrices 
\begin{align} \label{eq:product-test-states}
    \underset{\ket{\psi_A},\ket{\psi_B} \in \mathbb{C}^d}{\mathbb{E}} [\psi_A^{\otimes n} \otimes \psi_B^{\otimes n}] \quad \text{from} \quad \underset{\ket{\psi} \in \mathbb{C}^{d^2}}{\mathbb{E}} [\psi^{\otimes n}],
\end{align}
which we denote $\rho^{(n)}_{\rm prod}$ and $\rho^{(n)}_{\rm far}$, respectively. For a bipartite pure state, the maximum squared overlap with a product state equals its largest squared Schmidt coefficient $\lambda_{\rm max}(\psi)$. Thus, the distance to the set of product states is $\sqrt{1-\lambda_{\rm max}(\psi)}$. A Haar random state on $\mathbb{C}^d \tilde{\otimes} \mathbb{C}^d$ satisfies $\lambda_{\rm max}(\psi)=O(1/d)$ and is thus far from all product states with extremely high probability~\cite{hayden2006Aspects,beckey2025Product}. Therefore, a product tester could distinguish these states with high probability. Thus, proving that it is hard to distinguish the states in Eq.~\eqref{eq:product-test-states} using PPT-BOTH measurements will imply the desired sample complexity lower bound on PPT-BOTH product testers.

As mentioned above, the product testing case has a tensor structure within copies and between copies. We regard each copy on $\HC_{A_i} \tilde{\otimes} \HC_{B_i} \cong \mathbb{C}^d \tilde{\otimes} \mathbb{C}^d$ as one sample register. In this setting, $\Gamma_S$ denotes the partial transpose of both $A_i$ and $B_i$ for every $i \in S \subseteq [n].$ Given this slight modification, we can say that a two-outcome POVM is PPT-BOTH if $0 \preceq M^{\Gamma_S} \preceq \Id$ for every $S \subseteq [n].$ This class contains all potentially adaptive single-copy protocols that measure one bipartite sample at a time.

With these modifications in place, the proof sketch is very simple. We express the difference of states as $\rho^{(n)}_{\rm prod} - \tau_{d^2}^{\otimes n} - (\rho^{(n)}_{\rm far} - \tau_{d^2}^{\otimes n})$ and then use triangle inequality to bound upper bound the bias as 
\begin{align}
 \abs{\tr{M(\rho^{(n)}_{\rm prod} - \tau_{d^2}^{\otimes n})}} + \abs{\tr{M(\rho^{(n)}_{\rm far} - \tau_{d^2}^{\otimes n})}}. 
\end{align}
The second term is identical to the bound proved in Theorem~\ref{thm:eps-ind-PPT-purity-LB} under the replacement $d \rightarrow d^2$, yielding an upper bound on that term of the form $O(n^2/d^2)$. The first term is handled by generalizing our telescoping technique to this bipartite tensor structure. This yields the stronger bound $O(n^2/d)$, allowing us to conclude
\begin{align}
    \abs{\tr{M(\rho^{(n)}_{\rm prod} - \rho^{(n)}_{\rm far})}} \leq O\left(\frac{n^2}{d}\right), 
\end{align}
as desired. There exists product testing algorithm using only nonadaptive single-copy local measurements~\cite{beckey2025Product}. When $k=2$ and $\eps$ is constant, the sample complexity of their algorithm is $O(\sqrt{d})$, implying worst case optimality of that algorithm and strengthening the $\Omega(\sqrt{d})$ lower bound proved via the tree formalism and a more involved technical analysis~\cite{liu2025Separation,beckey2025Product}.

\lsec{Conclusion and outlook}
In this work, we introduced an elementary method for proving sample complexity lower bounds against adaptive single-copy protocols by relaxing to the larger class of PPT-BOTH measurements. Although this class extends strictly beyond separable measurements, the relaxation loses nothing asymptotically for purity testing and two-party product testing: in both cases, our lower bounds for all PPT-BOTH measurements are matched by nonadaptive single-copy protocols. These results establish PPT as a sharp tool for studying experimentally accessible quantum measurements.

Our proof uses only elementary linear algebra and exact identities for the symmetric subspace. Unlike the approximate-orthogonality approach~\cite{harrow2023Approximate}, which requires the local dimension to be large compared with the number of copies, our bounds hold for all $d,n\geq2$. The method therefore offers a complementary route to lower bounds when learning-tree analyses are technically difficult. Forthcoming work applies the same perspective to obtain stronger multipartite quantum-data-hiding bounds~\cite{akresh2026Optimal}.

A natural question arising in our work is whether our techniques can be applied to single-copy stabilizer testing. After encountering obstacles to a learning tree proof, Hinsche and Helsen adapted Harrow's approximate-orthogonality strategy to the Clifford commutant, obtaining an $\Omega(\sqrt{n})$ lower bound~\cite{hinsche2025SingleCopy}, falling short of the optimal $\Omega(n)$ bound that was recently established as a corollary of a more general result proved via the learning trees~\cite{arunachalam2026Optimal}. It would be interesting to determine whether our techniques yield a simpler proof of this single-copy lower bound. This open question points to a broader possibility: partial-transposition methods may offer a widely applicable route to sharp single-copy lower bounds in quantum testing and learning.

\lsec{Acknowledgments} Chat GPT 5.6 Sol was used as a research collaborator to iterate and refine the results. The authors take full responsibility for the correctness, exposition, and attribution in the final manuscript. The authors thank Louis Schatzki, Luke Coffman, Graeme Smith, Eric Chitambar, Felix Leditzky, and Fernando Geronimo for helpful discussions and input on the results. J.L.B is supported by a National Science Foundation Mathematical Sciences Postdoctoral Research Fellowship under Award No.~2402287 and an IQUIST Postdoctoral Fellowship.

\bibliography{main.bib}

@inproceedings{odonnell2016Efficient,
  title = {Efficient Quantum Tomography},
  booktitle = {Proceedings of the Forty-Eighth Annual {{ACM}} Symposium on {{Theory}} of {{Computing}}},
  author = {O'Donnell, Ryan and Wright, John},
  year = {2016},
  month = jun,
  series = {{{STOC}} '16},
  pages = {899--912},
  publisher = {Association for Computing Machinery},
  address = {New York, NY, USA},
  doi = {10.1145/2897518.2897544},
  isbn = {978-1-4503-4132-5}
}

@article{haah2017SampleOptimal,
  title = {Sample-{{Optimal Tomography}} of {{Quantum States}}},
  author = {Haah, Jeongwan and Harrow, Aram W. and Ji, Zhengfeng and Wu, Xiaodi and Yu, Nengkun},
  year = {2017},
  month = sep,
  journal = {IEEE Transactions on Information Theory},
  volume = {63},
  number = {9},
  pages = {5628--5641},
  issn = {1557-9654},
  doi = {10.1109/TIT.2017.2719044}
}

@misc{chen2024optimal,
  title = {An Optimal Tradeoff between Entanglement and Copy Complexity for State Tomography},
  author = {Chen, Sitan and Li, Jerry and Liu, Allen},
  year = {2024},
  month = feb,
  number = {arXiv:2402.16353},
  eprint = {2402.16353},
  primaryclass = {quant-ph},
  publisher = {arXiv},
  doi = {10.48550/arXiv.2402.16353},
  archiveprefix = {arxiv}
}

@article{eggeling2002Hiding,
  title = {Hiding {{Classical Data}} in {{Multipartite Quantum States}}},
  author = {Eggeling, T. and Werner, R. F.},
  year = 2002,
  month = aug,
  journal = {Physical Review Letters},
  volume = {89},
  number = {9},
  pages = {097905},
  publisher = {American Physical Society},
  doi = {10.1103/PhysRevLett.89.097905},
  note={\href{https://arxiv.org/abs/quant-ph/0203004}{arXiv:quant-ph/0203004}}
}

@article{huang2022Quantum,
  title = {Quantum Advantage in Learning from Experiments},
  author = {Huang, Hsin-Yuan and Broughton, Michael and Cotler, Jordan and Chen, Sitan and Li, Jerry and Mohseni, Masoud and Neven, Hartmut and Babbush, Ryan and Kueng, Richard and Preskill, John and McClean, Jarrod R.},
  year = {2022},
  month = jun,
  journal = {Science},
  volume = {376},
  number = {6598},
  pages = {1182--1186},
  publisher = {American Association for the Advancement of Science},
  doi = {10.1126/science.abn7293}
}

@inproceedings{chen2022Exponential,
  title = {Exponential {{Separations Between Learning With}} and {{Without Quantum Memory}}},
  booktitle = {2021 {{IEEE}} 62nd {{Annual Symposium}} on {{Foundations}} of {{Computer Science}} ({{FOCS}})},
  author = {Chen, Sitan and Cotler, Jordan and Huang, Hsin-Yuan and Li, Jerry},
  year = {2022},
  month = feb,
  pages = {574--585},
  issn = {2575-8454},
  doi = {10.1109/FOCS52979.2021.00063},
  note={\href{https://arxiv.org/abs/2111.05881}{arXiv:2111.05881}}
}

@misc{chen2022When,
  title = {When {{Does Adaptivity Help}} for {{Quantum State Learning}}?},
  author = {Chen, Sitan and Huang, Brice and Li, Jerry and Liu, Allen and Sellke, Mark},
  year = {2022},
  month = jun,
  journal = {arXiv.org},
  howpublished = {https://arxiv.org/abs/2206.05265v2},
  langid = {english}
}

@article{hayden2006Aspects,
  title = {Aspects of {{Generic Entanglement}}},
  author = {Hayden, Patrick and Leung, Debbie W. and Winter, Andreas},
  year = {2006},
  month = jul,
  journal = {Communications in Mathematical Physics},
  volume = {265},
  number = {1},
  pages = {95--117},
  issn = {1432-0916},
  doi = {10.1007/s00220-006-1535-6},
  langid = {english}
}

@article{guhne2009Entanglement,
  title = {Entanglement Detection},
  author = {G{\"u}hne, Otfried and T{\'o}th, G{\'e}za},
  year = {2009},
  month = apr,
  journal = {Physics Reports},
  volume = {474},
  number = {1},
  pages = {1--75},
  issn = {0370-1573},
  doi = {10.1016/j.physrep.2009.02.004}
}

@article{chitambar2014Everything,
  title = {Everything {{You Always Wanted}} to {{Know About LOCC}} ({{But Were Afraid}} to {{Ask}})},
  author = {Chitambar, Eric and Leung, Debbie and Man{\v c}inska, Laura and Ozols, Maris and Winter, Andreas},
  year = {2014},
  month = may,
  journal = {Communications in Mathematical Physics},
  volume = {328},
  number = {1},
  pages = {303--326},
  issn = {1432-0916},
  doi = {10.1007/s00220-014-1953-9},
  langid = {english},
  note={\href{https://arxiv.org/abs/1210.4583}{arXiv:1210.4583}}
}

@article{liu2025Separation,
  title = {Separation between Entanglement Criteria and Entanglement Detection Protocols},
  author = {Liu, Zhenhuan and Wei, Fuchuan},
  year = 2025,
  month = aug,
  journal = {Physical Review Research},
  volume = {7},
  number = {3},
  pages = {033121},
  issn = {2643-1564},
  doi = {10.1103/hm8j-wgqm},
  note={\href{https://arxiv.org/abs/2403.01664}{arXiv:2403.01664}}
}

@article{montanaro2016Survey,
  title = {A {{Survey}} of {{Quantum Property Testing}}},
  author = {Montanaro, Ashley and De Wolf, Ronald},
  year = {2016},
  journal = {Theory of Computing},
  volume = {1},
  number = {1},
  pages = {1--81},
  issn = {1557-2862},
  doi = {10.4086/toc.gs.2016.007},
  langid = {english}
}

@misc{gong2024sample,
  title = {On the Sample Complexity of Purity and Inner Product Estimation},
  author = {Gong, Weiyuan and Haferkamp, Jonas and Ye, Qi and Zhang, Zhihan},
  year = {2024},
  month = oct,
  number = {arXiv:2410.12712},
  eprint = {2410.12712},
  primaryclass = {quant-ph},
  publisher = {arXiv},
  archiveprefix = {arXiv},
  langid = {english},
  note={\href{https://arxiv.org/abs/2410.12712}{arXiv:2410.12712}}
}

@misc{beckey2025Product,
  title = {Product Testing with Single-Copy Measurements},
  author = {Beckey, Jacob and Coffman, Luke and Shlosberg, Ariel and Schatzki, Louis and Leditzky, Felix},
  year = 2025,
  month = oct,
  number = {arXiv:2510.07820},
  eprint = {2510.07820},
  primaryclass = {quant-ph},
  publisher = {arXiv},
  doi = {10.48550/arXiv.2510.07820},
  archiveprefix = {arXiv},
  note={\href{https://arxiv.org/abs/2510.07820}{arXiv:2510.07820}}
}

@misc{arunachalam2026Optimal,
  title = {Optimal {{Stabilizer Testing}} and {{Learning}} with {{Limited Quantum Memory}}},
  author = {Arunachalam, Srinivasan and Schatzki, Louis},
  year = 2026,
  month = jul,
  number = {arXiv:2607.02444},
  eprint = {2607.02444},
  primaryclass = {quant-ph},
  publisher = {arXiv},
  doi = {10.48550/arXiv.2607.02444},
  archiveprefix = {arXiv},
  note={\href{https://arxiv.org/abs/2607.02444}{arXiv:2607.02444}}
}

@article{harrow2023Approximate,
  title = {Approximate Orthogonality of Permutation Operators, with Application to Quantum Information},
  author = {Harrow, Aram W.},
  year = {2023},
  month = dec,
  journal = {Letters in Mathematical Physics},
  volume = {114},
  number = {1},
  eprint = {2309.00715},
  primaryclass = {quant-ph},
  pages = {1},
  issn = {1573-0530},
  doi = {10.1007/s11005-023-01744-1},
  archiveprefix = {arXiv},
  note={\href{https://arxiv.org/abs/2309.00715}{arXiv:2309.00715}}
}

@article{horodecki2009Quantum,
  title = {Quantum Entanglement},
  author = {Horodecki, Ryszard and Horodecki, Pawe{\l} and Horodecki, Micha{\l} and Horodecki, Karol},
  year = {2009},
  month = jun,
  journal = {Reviews of Modern Physics},
  volume = {81},
  number = {2},
  pages = {865--942},
  issn = {0034-6861, 1539-0756},
  doi = {10.1103/RevModPhys.81.865},
  langid = {english}
}

@article{bluvstein2022quantum,
  title = {A Quantum Processor Based on Coherent Transport of Entangled Atom Arrays},
  author = {Bluvstein, Dolev and Levine, Harry and Semeghini, Giulia and Wang, Tout T. and Ebadi, Sepehr and Kalinowski, Marcin and Keesling, Alexander and Maskara, Nishad and Pichler, Hannes and Greiner, Markus and Vuleti{\'c}, Vladan and Lukin, Mikhail D.},
  year = {2022},
  month = apr,
  journal = {Nature},
  volume = {604},
  number = {7906},
  pages = {451--456},
  publisher = {Nature Publishing Group},
  issn = {1476-4687},
  doi = {10.1038/s41586-022-04592-6},
  copyright = {2022 The Author(s)},
  langid = {english}
}

@article{bluvstein2023Logical,
  title = {Logical Quantum Processor Based on Reconfigurable Atom Arrays},
  author = {Bluvstein, Dolev and Evered, Simon J. and Geim, Alexandra A. and Li, Sophie H. and Zhou, Hengyun and Manovitz, Tom and Ebadi, Sepehr and Cain, Madelyn and Kalinowski, Marcin and Hangleiter, Dominik and Ataides, J. Pablo Bonilla and Maskara, Nishad and Cong, Iris and Gao, Xun and Rodriguez, Pedro Sales and Karolyshyn, Thomas and Semeghini, Giulia and Gullans, Michael J. and Greiner, Markus and Vuleti{\'c}, Vladan and Lukin, Mikhail D.},
  year = {2023},
  month = dec,
  journal = {Nature},
  pages = {1--3},
  publisher = {Nature Publishing Group},
  issn = {1476-4687},
  doi = {10.1038/s41586-023-06927-3},
  copyright = {2023 The Author(s), under exclusive licence to Springer Nature Limited},
  langid = {english}
}

@misc{harrow2013Church,
  title = {The {{Church}} of the {{Symmetric Subspace}}},
  author = {Harrow, Aram W.},
  year = {2013},
  month = aug,
  number = {arXiv:1308.6595},
  eprint = {1308.6595},
  publisher = {arXiv},
  doi = {10.48550/arXiv.1308.6595},
  archiveprefix = {arXiv},
  note = {\href{https://arxiv.org/abs/1308.6595}{arXiv:1308.6595}}
}

@inproceedings{bubeck2020Entanglement,
  title = {Entanglement Is {{Necessary}} for {{Optimal Quantum Property Testing}}},
  booktitle = {2020 {{IEEE}} 61st {{Annual Symposium}} on {{Foundations}} of {{Computer Science}} ({{FOCS}})},
  author = {Bubeck, Sebastien and Chen, Sitan and Li, Jerry},
  year = {2020},
  month = nov,
  pages = {692--703},
  issn = {2575-8454},
  doi = {10.1109/FOCS46700.2020.00070}
}

@article{ekert2002Direct,
  title = {Direct {{Estimations}} of {{Linear}} and {{Nonlinear Functionals}} of a {{Quantum State}}},
  author = {Ekert, Artur K. and Alves, Carolina Moura and Oi, Daniel K. L. and Horodecki, Micha{\l} and Horodecki, Pawe{\l} and Kwek, L. C.},
  year = {2002},
  month = may,
  journal = {Physical Review Letters},
  volume = {88},
  number = {21},
  pages = {217901},
  publisher = {American Physical Society},
  doi = {10.1103/PhysRevLett.88.217901}
}

@article{barenco1997Stabilization,
  title = {Stabilization of {{Quantum Computations}} by {{Symmetrization}}},
  author = {Barenco, Adriano and Berthiaume, Andr{\'e} and Deutsch, David and Ekert, Artur and Jozsa, Richard and Macchiavello, Chiara},
  year = {1997},
  month = oct,
  journal = {SIAM Journal on Computing},
  volume = {26},
  number = {5},
  pages = {1541--1557},
  publisher = {{Society for Industrial and Applied Mathematics}},
  issn = {0097-5397},
  doi = {10.1137/S0097539796302452},
  note={\href{https://arxiv.org/abs/quant-ph/9604028}{arXiv:quant-ph/9604028}}
}

@inproceedings{hinsche2025SingleCopy,
  title = {Single-{{Copy Stabilizer Testing}}},
  booktitle = {Proceedings of the 57th {{Annual ACM Symposium}} on {{Theory}} of {{Computing}}},
  author = {Hinsche, Marcel and Helsen, Jonas},
  year = {2025},
  month = jun,
  series = {{{STOC}} '25},
  pages = {439--450},
  publisher = {Association for Computing Machinery},
  address = {New York, NY, USA},
  doi = {10.1145/3717823.3718169},
  isbn = {979-8-4007-1510-5},
  note={\href{https://arxiv.org/abs/2410.07986}{arXiv:2410.07986}}
}

@misc{bakshi2024Learning,
  title = {Learning the Closest Product State},
  author = {Bakshi, Ainesh and Bostanci, John and Kretschmer, William and Landau, Zeph and Li, Jerry and Liu, Allen and O'Donnell, Ryan and Tang, Ewin},
  year = 2024,
  month = nov,
  number = {arXiv:2411.04283},
  eprint = {2411.04283},
  primaryclass = {quant-ph},
  publisher = {arXiv},
  doi = {10.48550/arXiv.2411.04283},
  archiveprefix = {arXiv}
}

@article{beigi2010Approximating,
  title = {Approximating the Set of Separable States Using the Positive Partial Transpose Test},
  author = {Beigi, Salman and Shor, Peter W.},
  year = 2010,
  month = apr,
  journal = {Journal of Mathematical Physics},
  volume = {51},
  number = {4},
  pages = {042202},
  issn = {0022-2488},
  doi = {10.1063/1.3364793}
}

@article{daguerre2025Experimental,
  title = {Experimental {{Demonstration}} of {{High-Fidelity Logical Magic States}} from {{Code Switching}}},
  author = {Daguerre, Lucas and {Blume-Kohout}, Robin and Brown, Natalie C. and Hayes, David and Kim, Isaac H.},
  year = 2025,
  month = oct,
  journal = {Physical Review X},
  volume = {15},
  number = {4},
  pages = {041008},
  publisher = {American Physical Society},
  doi = {10.1103/dck4-x9c2}
}

@unpublished{akresh2026Optimal,
    author = {Oren Akresh and Felix Leditzky and Jacob Beckey},
    title  = {Optimal Quantum Data Hiding with Multipartite Werner States},
    note   = {In preparation},
    year   = {2026}
}

@article{mintert2005Concurrence,
  title = {Concurrence of {{Mixed Multipartite Quantum States}}},
  author = {Mintert, Florian and Ku{\'s}, Marek and Buchleitner, Andreas},
  year = 2005,
  month = dec,
  journal = {Physical Review Letters},
  volume = {95},
  number = {26},
  pages = {260502},
  publisher = {American Physical Society},
  doi = {10.1103/PhysRevLett.95.260502}
}

@article{bennett1999Unextendible,
  title = {Unextendible {{Product Bases}} and {{Bound Entanglement}}},
  author = {Bennett, Charles H. and DiVincenzo, David P. and Mor, Tal and Shor, Peter W. and Smolin, John A. and Terhal, Barbara M.},
  year = 1999,
  month = jun,
  journal = {Physical Review Letters},
  volume = {82},
  number = {26},
  pages = {5385--5388},
  publisher = {American Physical Society},
  doi = {10.1103/PhysRevLett.82.5385}
}

@article{brennen2003observable,
  title = {An Observable Measure of Entanglement for Pure States of Multi-Qubit Systems},
  author = {Brennen, Gavin K.},
  year = 2003,
  month = nov,
  journal = {Quantum Info. Comput.},
  volume = {3},
  number = {6},
  pages = {619--626},
  issn = {1533-7146}
}

@article{horodecki1997Separability,
  title = {Separability Criterion and Inseparable Mixed States with Positive Partial Transposition},
  author = {Horodecki, Pawel},
  year = 1997,
  month = aug,
  journal = {Physics Letters A},
  volume = {232},
  number = {5},
  pages = {333--339},
  issn = {0375-9601},
  doi = {10.1016/S0375-9601(97)00416-7}
}

@article{liu2023Complexity,
  title = {The {{Round Complexity}} of {{Local Operations}} and {{Classical Communication}} ({{LOCC}}) in {{Random-Party Entanglement Distillation}}},
  author = {Liu, Guangkuo and George, Ian and Chitambar, Eric},
  year = 2023,
  month = sep,
  journal = {Quantum},
  volume = {7},
  pages = {1104},
  publisher = {Verein zur F\"orderung des Open Access Publizierens in den Quantenwissenschaften},
  doi = {10.22331/q-2023-09-07-1104},
  langid = {british}
}

@article{cheng2023Discrimination,
  title = {Discrimination of {{Quantum States Under Locality Constraints}} in the {{Many-Copy Setting}}},
  author = {Cheng, Hao-Chung and Winter, Andreas and Yu, Nengkun},
  year = 2023,
  month = nov,
  journal = {Communications in Mathematical Physics},
  volume = {404},
  number = {1},
  pages = {151--183},
  issn = {1432-0916},
  doi = {10.1007/s00220-023-04836-0},
  langid = {english}
}

@article{bennett1999Quantum,
  title = {Quantum Nonlocality without Entanglement},
  author = {Bennett, Charles H. and DiVincenzo, David P. and Fuchs, Christopher A. and Mor, Tal and Rains, Eric and Shor, Peter W. and Smolin, John A. and Wootters, William K.},
  year = 1999,
  month = feb,
  journal = {Physical Review A},
  volume = {59},
  number = {2},
  pages = {1070--1091},
  publisher = {American Physical Society},
  doi = {10.1103/PhysRevA.59.1070}
}

@article{miller2026Experimental,
  title = {Experimental Measurement and a Physical Interpretation of Quantum Shadow Enumerators},
  author = {Miller, Daniel and Levi, Kyano and Postler, Lukas and Steiner, Alex and Bittel, Lennart and White, Gregory A. L. and Tang, Yifan and Kuehnke, Eric J. and Mele, Antonio A. and Khatri, Sumeet and Leone, Lorenzo and Carrasco, Jose and Marciniak, Christian D. and Pogorelov, Ivan and {Guevara-Bertsch}, Milena and Freund, Robert and Blatt, Rainer and Schindler, Philipp and Monz, Thomas and Ringbauer, Martin and Eisert, Jens},
  year = 2026,
  month = jun,
  journal = {Physical Review Research},
  volume = {8},
  number = {2},
  pages = {023318},
  publisher = {American Physical Society},
  doi = {10.1103/8h3b-brg1}
}

@article{matthews2009Distinguishability,
  title = {Distinguishability of {{Quantum States Under Restricted Families}} of {{Measurements}} with an {{Application}} to {{Quantum Data Hiding}}},
  author = {Matthews, William and Wehner, Stephanie and Winter, Andreas},
  year = 2009,
  month = nov,
  journal = {Communications in Mathematical Physics},
  volume = {291},
  number = {3},
  pages = {813--843},
  issn = {1432-0916},
  doi = {10.1007/s00220-009-0890-5},
  langid = {english}
}

@article{rains2001semidefinite,
  title = {A Semidefinite Program for Distillable Entanglement},
  author = {Rains, E.M.},
  year = 2001,
  month = nov,
  journal = {IEEE Transactions on Information Theory},
  volume = {47},
  number = {7},
  pages = {2921--2933},
  issn = {1557-9654},
  doi = {10.1109/18.959270}
}

@article{horodecki1996Separability,
  title = {Separability of Mixed States: Necessary and Sufficient Conditions},
  shorttitle = {Separability of Mixed States},
  author = {Horodecki, Micha{\l} and Horodecki, Pawe{\l} and Horodecki, Ryszard},
  year = 1996,
  month = nov,
  journal = {Physics Letters A},
  volume = {223},
  number = {1},
  pages = {1--8},
  issn = {0375-9601},
  doi = {10.1016/S0375-9601(96)00706-2}
}

@article{peres1996Separability,
  title = {Separability {{Criterion}} for {{Density Matrices}}},
  author = {Peres, Asher},
  year = 1996,
  month = aug,
  journal = {Physical Review Letters},
  volume = {77},
  number = {8},
  pages = {1413--1415},
  issn = {0031-9007, 1079-7114},
  doi = {10.1103/PhysRevLett.77.1413},
  langid = {english}
}

@misc{noller2025infinite,
  title = {An Infinite Hierarchy of Multi-Copy Quantum Learning Tasks},
  author = {N{\"o}ller, Jan and Tran, Viet T. and Gachechiladze, Mariami and Kueng, Richard},
  year = 2025,
  publisher = {arXiv},
  doi = {10.48550/ARXIV.2510.08070},
  copyright = {Creative Commons Attribution 4.0 International}
}

@misc{ye2025Exponential,
  title = {Exponential {{Advantage}} from {{One More Replica}} in {{Estimating Nonlinear Properties}} of {{Quantum States}}},
  author = {Ye, Qi and Liu, Zhenhuan and Deng, Dong-Ling},
  year = 2025,
  publisher = {arXiv},
  doi = {10.48550/ARXIV.2509.24000},
  copyright = {Creative Commons Attribution 4.0 International}
}

@inproceedings{barak2012Hypercontractivity,
  title = {Hypercontractivity, Sum-of-Squares Proofs, and Their Applications},
  booktitle = {Proceedings of the Forty-Fourth Annual {{ACM}} Symposium on {{Theory}} of Computing},
  author = {Barak, Boaz and Brandao, Fernando G.S.L. and Harrow, Aram W. and Kelner, Jonathan and Steurer, David and Zhou, Yuan},
  year = 2012,
  month = may,
  series = {{{STOC}} '12},
  pages = {307--326},
  publisher = {Association for Computing Machinery},
  address = {New York, NY, USA},
  doi = {10.1145/2213977.2214006},
  isbn = {978-1-4503-1245-5}
}

@article{brandao2016ProductState,
  title = {Product-{{State Approximations}} to {{Quantum States}}},
  author = {Brand{\~a}o, Fernando G. S. L. and Harrow, Aram W.},
  year = 2016,
  month = feb,
  journal = {Communications in Mathematical Physics},
  volume = {342},
  number = {1},
  pages = {47--80},
  issn = {1432-0916},
  doi = {10.1007/s00220-016-2575-1},
  langid = {english}
}

@article{jeronimo2026QMA2,
  title = {The {{QMA}}(2) {{UniverseComplexity}}, {{Entanglement}}, and {{Optimization}}},
  author = {Jeronimo, Fernando Granha and Wu, Pei and Leigh, Itai},
  year = 2026,
  month = mar,
  journal = {ACM SIGACT News},
  volume = {57},
  number = {1},
  pages = {64--99},
  issn = {0163-5700},
  doi = {10.1145/3802807.3802814}
}

@inproceedings{harrow_testing_2010,
    title = {Testing product states, quantum {Merlin}-{Arthur} games and tensor optimisation},
    url = {http://arxiv.org/abs/1001.0017},
    doi = {10.1109/FOCS.2010.66 10.1145/2432622.2432625},
    urldate = {2025-04-03},
    booktitle = {2010 {IEEE} 51st {Annual} {Symposium} on {Foundations} of {Computer} {Science}},
    author = {Harrow, Aram W. and Montanaro, Ashley},
    month = oct,
    year = {2010},
    note = {arXiv:1001.0017 [quant-ph]},
    pages = {633--642},
}

@misc{chen2024Optimala,
  title = {Optimal Tradeoffs for Estimating {{Pauli}} Observables},
  author = {Chen, Sitan and Gong, Weiyuan and Ye, Qi},
  year = {2024},
  month = apr,
  number = {arXiv:2404.19105},
  eprint = {2404.19105},
  publisher = {arXiv},
  doi = {10.48550/arXiv.2404.19105},
  archiveprefix = {arXiv}
}
\onecolumngrid
\section{End Matter}
\twocolumngrid
\lsec{Preliminaries} Let $\SC_n$ denote the symmetric group on $n$ objects. We denote the unitary representation of $\SC_n$ on $(\mathbb{C}^d)^{\otimes n}$ by $P_d(\pi)$. Crucially, because it is a representation, we have $P_d(\pi \sigma) =P_d(\pi) P_d(\sigma)$ and $P_d(e) = \Id_d^{\otimes n}$, where $e\in \SC_n$ is the identity permutation. Let $\vee^{n}(\mathbb{C}^d) \subseteq (\mathbb{C}^d)^{\otimes n}$ denote the symmetric subspace. The projector onto this subspace can be expressed as (cf. Ref.~\cite{harrow2013Church})
\begin{align}
\Pi_{\mathrm{sym}}^{d,n} = \frac{1}{n!} \sum_{\pi \in \SC_n} P_d(\pi),
\end{align}
with the dimension of this subspace being given as
\begin{align}
     D_n := \dim\left(\vee^{n}(\mathbb{C}^d)\right) = \binom{d+n-1}{n}.
\end{align}
When the local dimension is clear from context, we let $\Pi_n \coloneq \Pi_{\mathrm{sym}}^{d,n}$ to ease notation. We define the maximally mixed state on the symmetric subspace as $\sigma^{(n)} \eqcolon \Pi_n/D_n$. Let the group of $d$-by-$d$ unitary matrices be denoted $\UC_d$. The Haar measure is the unique left- and right-invariant measure on this group. We define a Haar random state to be $\ket{\psi} = U \ket{0}$, where $U$ is sampled according to the Haar measure on $\UC_d$.  

The average of $n$ copies of a Haar random pure state yields the maximally mixed state on the symmetric subspace
\begin{align}
        \underset{U \in \UC_d}{\mathbb{E}}\left[U^{\otimes n} \ketbra{0}{0}(U^\dagger)^{\otimes n}\right] = \sigma^{(n)}.
    \end{align}
The partial transposition operation, denoted $\Gamma_S$, transposes indices $i \in S \subseteq [n]$ and leaves the others unchanged. Let $X,Y$ be bounded operators on $(\mathbb{C}^d)^{\otimes n}$. Expanding in an operator basis, one can verify $\tr{X^{\Gamma_S}} = \tr{X}$ and $\tr{X^{\Gamma_S} Y^{\Gamma_S}}= \tr{XY}$. 

We denote the SWAP operator as $\mathbb{F} = P_d((ij))$, where $(ij)$ is cycle notation indicating which indices are being swapped. A well-known fact from quantum information theory is that partially transposing half of the SWAP operator yields $d$ times the maximally entangled projector. In our notation, we write
\begin{align}
    P_d((ij))^{\Gamma_{\{i\}}} = P_d((ij))^{\Gamma_{\{j\}}} = d \psi_{i,j}.
\end{align}
Let $M \succeq 0$. Then a two-outcome POVM $\{M,\Id-M\}$ on $(\mathbb{C}^d)^{\otimes n}$ is said to be PPT-BOTH if both $M$ and $\Id-M$ remain positive after partial transposition of any subset of indices. These conditions can be succinctly summarized as
\begin{align} 
    0 \preceq M^{\Gamma_S} \preceq \Id \quad \forall~S \subseteq [n].
\end{align}
Let $A,B$ be bounded operators. The trace norm is denoted $\|A\|_1 = \tr{\sqrt{A^{\dagger}A}}$ and, when $A \succeq 0$, simplifies to $\|A\|_1=\tr{A}$. The trace norm is multiplicative over tensor products $\|A \otimes B \|_1 = \|A\|_1 \|B\|_1.$ With these facts in mind, we make two observations. First, if $X,Y \succeq 0$ and $[X,Y]=0,$ then $XY \succeq 0$. Second, if $X$ satisfies $0 \preceq X \preceq \Id$, and $Y$ is traceless and Hermitian, then $\tr{XY} \leq \frac{1}{2}\|Y\|_1$. This follows from decomposing $Y=Y_+ - Y_-$, where $Y_+, Y_- \succeq 0$ and noting that tracelessness implies $\tr{Y_+}=\tr{Y_-}=\|Y\|_1/2$. With our notation and basic facts in mind, we proceed to the proof of our main results. See the Supplemental Material for more detailed proofs, relations to prior work, and a discussion of the algorithms that saturate our lower bounds.

\lsec{Purity Testing}
As outlined in the main text, the first step in our proof is rewriting the difference of the two states using a telescoping sum. Noting that $\sigma^{(1)} = \tau_d$, we are able to write
    \begin{align}
        \Delta 
        &= \sum_{m=2}^n (\sigma^{(m)} - \sigma^{(m-1)}\otimes \tau_d)\otimes \tau_d^{\otimes (n-m)},\\
        &=  \sum_{m=2}^n \DC_m(\sigma^{(m)}) \otimes \tau_d^{\otimes (n-m)},
    \end{align}
    where we have defined $\DC_m(\sigma^{(m)}) \coloneq \sigma^{(m)} - \sigma^{(m-1)}\otimes \tau_d$, which we note is traceless and Hermitian. We may then bound the bias as
    \begin{align}
        \tr{M \Delta} &= \sum_{m=2}^n \tr{M \DC_m(\sigma^{(m)}) \otimes \tau_d^{\otimes (n-m)}}, \\
        &= \sum_{m=2}^n \tr{M^{\Gamma_m}\DC_m(\sigma^{(m)})^{\Gamma_m} \otimes \tau_d^{\otimes (n-m)} }, \\
        &\leq \sum_{m=2}^n \frac{1}{2} \left\|\DC_m(\sigma^{(m)})^{\Gamma_m} \right\|_1, 
    \end{align}
    where we have used the PPT-BOTH condition to ensure $0 \preceq M^{\Gamma_m} \preceq \Id$. This, coupled with the Hermiticity and tracelessness of $\DC_{m}(\sigma^{(m)})^{\Gamma_m}$ allows us to apply the trace norm lemma from the preliminaries. All that remains is to prove that 
    \begin{align}
        \frac{1}{2} \left\|\DC_m(\sigma^{(m)})^{\Gamma_m} \right\|_1 \leq \frac{m-1}{d+m-1},
    \end{align}
    from which the $O(n^2/d)$ bound follows simply. 

    \begin{lemma}\label{lem:last-leg-depolarize-trace-norm-bound}
    Let $\DC_m(\sigma^{(m)})= (\sigma^{(m)} - \sigma^{(m-1)}\otimes \tau_d)$ and denote the partial transposition of the $m$-th leg as $\Gamma_m$. Then, we have \begin{align}
        \frac{1}{2}\|\DC_m(\sigma^{(m)})^{\Gamma_m}\|_1 \leq \frac{m-1}{d+m-1}.
    \end{align}
\end{lemma}
\begin{proof}
    To begin, expand the symmetric subspace projector and isolate the last tensor factor \begin{align}
         \Pi_m  &= \frac{1}{m!}\sum_{\pi \in S_m}P_d(\pi), \\
        & = \frac{1}{m}\left(\Pi_{m-1} \otimes \mathbb{I}_m\right)\left(\mathbb{I} + \sum_{i=1}^{m-1}P_d((im))\right).
    \end{align}
    Taking the partial transposition of the $m$-th site yields 
    \begin{align}
          \Pi_m^{\Gamma_m}& =  \frac{1}{m}\left(\Pi_{m-1} \otimes \mathbb{I}_m\right)\left(\mathbb{I} + \sum_{i=1}^{m-1}d \psi_{i,m}\right).
    \end{align}
   To ease the remaining manipulation, define the positive operators \begin{align}
        A_m \coloneq \Pi_{m-1} \otimes \mathbb{I}_m \quad \text{and} \quad \Psi_m\coloneq \sum_{i=1}^{m-1}\psi_{i,m},
    \end{align}
    which allows us to rewrite \begin{align}
        \Pi_m^{\Gamma_m} = \frac{1}{m}A_m(\mathbb{I}+d \Psi_m).
    \end{align}
    Next, let $\pi' \in S_{m-1}$ be an arbitrary permutation of the first $m-1$ sites. Observing that $[P_d(\pi'),\Psi_m] =0$, linearity allows us to conclude that $[A_m,\Psi_m] = 0$. Thus $0 \preceq A_m\Psi_m$, which in turn implies $\|A_m\Psi_m\|_1 = \tr{A_m\Psi_m}$. With this in mind, we note that $\|A_m\|_1 = d D_{m-1}$ and $\|A_m\Psi_m\|_1 = (m-1)D_{m-1}/d.$

    With all of these facts, we may write the partially transposed single-erasure operator as \begin{align}
        &\DC_m(\sigma^{(m)})^{\Gamma_m} = \frac{1}{D_m}\Pi_m^{\Gamma_m} - \frac{1}{dD_{m-1}}\Pi_{m-1}\otimes \mathbb{I},\\
        &= \left(\frac{1}{mD_m}-\frac{1}{d D_{m-1}}\right)A_m + \frac{d}{m D_m}A_m\Psi_m
    \end{align}
    A final application of the triangle inequality and division by $2$ yields \begin{align}
        \frac{1}{2}\|\DC_m(\sigma^{(m)})^{\Gamma_m}\|_1 \leq \frac{m-1}{d+m-1},
    \end{align}
    as desired. Finally, note that the theorem statement $O(n^2/d)$ is trivial unless $n=O(\sqrt{d})$. In this regime, we may upper bound $\frac{m-1}{d+m-1} \leq \frac{m-1}{d}$. Summing from $m=2$ to $n$ yields the desired $O(n^2/d)$ upper bound. Applying the same argument to $\Id-M$ and using $\tr{\Delta}=0$ yields the two-sided bound.
\end{proof}

\lsec{Product Testing}
As mentioned in the main text, our task of distinguishing $\rho_{\mathrm{prod}}^{(n)}$ and $\rho_{\mathrm{far}}^{(n)}$ reduces to showing \begin{align}
    \abs{\tr{M(\rho_{\mathrm{prod}}^{(n)} - \tau_{d^2}^{\otimes n})}} \leq O\left(\frac{n^2}{d}\right).
\end{align}
Similar to the purity testing proof, we can rewrite the difference of these two states using a telescoping sum. Recall that $\rho_{\mathrm{prod}}^{(m)} = \sigma_A^{(m)}\tilde{\otimes} \sigma_B^{(m)}$, where $\sigma_A^{(m)}$ and $\sigma_B^{(m)}$ denote the maximally mixed states on the symmetric subspace of Alice and Bob's respective systems. It follows that the difference of states can be written as 
\begin{align}
\begin{split}
&\rho_{\mathrm{prod}}^{(n)} - \tau_{d^2}^{\otimes n} \\
&= \sum_{m=2}^n \left[\rho_{\mathrm{prod}}^{(m)} - \rho_{\mathrm{prod}}^{(m-1)} \otimes \tau_{d^2}\right]\otimes \tau_{d^2}^{\otimes(n-m)}, \\
& = \sum_{m=2}^n\left[\DC_m(\sigma_A^{(m)}) \tilde{\otimes} \sigma_B^{(m)}\right] \otimes \tau_{d^2}^{\otimes (n-m)} \\
&+ \sum_{m=2}^n\left[(\sigma_A^{(m-1)} \otimes \tau_d) \tilde{\otimes}\DC_m(\sigma_B^{(m)})\right]\otimes \tau_{d^2}^{\otimes (n-m)}.
\end{split}
\end{align}
The bias can then be upper bounded using Theorem~\ref{thm:eps-ind-PPT-purity-LB}, yielding \begin{align}
    &\abs{\tr{M(\rho_{\mathrm{prod}}^{(n)} - \tau_{d^2}^{\otimes n})}}  \\
    & \leq  \frac{1}{2}\sum_{m=2}^n\left\|\DC_m(\sigma_A^{(m)})^{\Gamma_m} \right\|_1 + \left\|\DC_m(\sigma_B^{(m)})^{\Gamma_m}\right\|_1, \\
    & \leq 2 \sum_{m=2}^n \frac{m-1}{d}, \\
    & = O\left(\frac{n^2}{d}\right).
\end{align}
\newpage
\onecolumngrid
\appendix
\section{Supplemental Material}
\subsection{Purity Testing Details}
\begin{theorem}[Theorem~\ref{thm:eps-ind-PPT-purity-LB} Restated]
For every $d,n\geq2$ and every PPT-BOTH measurement $\{M,\Id-M\}$,
\begin{align}
\abs{
\tr{
M\left(
\sigma^{(n)}-\tau_d^{\otimes n}
\right)
}
}
\leq
\frac{n(n-1)}{2d}.
\end{align}
Consequently, constant-bias purity testing under PPT-BOTH measurements requires $n=\Omega(\sqrt d)$ samples.
\end{theorem}
\begin{proof}
    First, note that any $\eps$-purity tester could certainly distinguish between a Haar random pure state and the maximally mixed state. Thus, if we can show that a Haar random pure state is difficult to distinguish from the maximally mixed state using only PPT-BOTH measurements, this will imply hardness of purity testing. 

    It is helpful to imagine a referee preparing, with equal probability, either the maximally mixed state $\tau_d$ or a Haar random pure state $\psi$. In either case, we will get $n$ identical copies of this state. Because we do not know what Haar random pure state we are given $n$ copies of, the density matrix representing our knowledge of the underlying system in this case is 
    \begin{align}
       \underset{U \in \UC_d}{\mathbb{E}}\left[U^{\otimes n} \ketbra{0}{0}(U^\dagger)^{\otimes n}\right] = \frac{\Pi_n}{D_n} \eqcolon \sigma^{(n)},
    \end{align}
    where we have introduced the shorthand $\Pi_n \coloneq \Pi_{\rm sym}^{d,n}$ and $D_n \coloneq \tr{\Pi_{\rm sym}^{d,n}}$, which will ease the algebraic manipulations below. Thus, distinguishing $n$ copies of a Haar random state from $n$ copies of the maximally mixed state is equivalent to distinguishing the maximally mixed state on $(\mathbb{C}^d)^{\otimes n}$ from the maximally mixed state on the symmetric subspace $\vee^n (\mathbb{C}^d)$.

    Let $\Delta \coloneq \sigma^{(n)} - \tau_d^{\otimes n}$ denote the difference of the two density operators. The assumption that a PPT-BOTH distinguisher exists is equivalent to assuming an $\Omega(1)$ lower bound on the bias. If we can establish a $O(n^2/d)$ upper bound on the bias, we may conclude 
    \begin{align}
        \Omega(1) \leq \abs{\tr{M\Delta}} \leq O\left(\frac{n^2}{d}\right) \implies n = \Omega(\sqrt{d}).
    \end{align}
    As outlined in the main text, the first step in our proof is rewriting the difference of the two states using a telescoping sum. Noting that $\sigma^{(1)} = \tau_d$, we are able to write
    \begin{align}
        \Delta &= \sigma^{(n)} - \sigma^{(n-1)}\otimes \tau_d + \sigma^{(n-1)}\otimes \tau_d - \dotsm - \sigma^{(2)} \otimes \tau_d^{\otimes (n-2)} + \sigma^{(2)} \otimes \tau_d^{\otimes (n-2)} - \tau_d^{\otimes n},\\
        &= \sum_{m=2}^n (\sigma^{(m)} - \sigma^{(m-1)}\otimes \tau_d)\otimes \tau_d^{\otimes (n-m)},\\
        &=  \sum_{m=2}^n \DC_m(\sigma^{(m)}) \otimes \tau_d^{\otimes (n-m)},
    \end{align}
    where we have defined $\DC_m(\sigma^{(m)}) \coloneq \sigma^{(m)} - \sigma^{(m-1)}\otimes \tau_d$, which we refer to as the \textit{single-depolarization operator}. Returning to bound the magnitude of the bias, we write
    \begin{align}
        \abs{\tr{M \Delta}} &= \abs{\sum_{m=2}^n \tr{M \DC_m(\sigma^{(m)}) \otimes \tau_d^{\otimes (n-m)}}}, \quad &\text{linearity of trace}\\
        &\leq \sum_{m=2}^n \abs{\tr{M\DC_m(\sigma^{(m)}) \otimes \tau_d^{\otimes (n-m)} }}, \quad &\text{triangle inequality}\\
        &= \sum_{m=2}^n \abs{\tr{M^{\Gamma_m}\DC_m(\sigma^{(m)})^{\Gamma_m} \otimes \tau_d^{\otimes (n-m)} }}, \quad \\
        &\leq \sum_{m=2}^n \frac{1}{2} \left\|\DC_m(\sigma^{(m)})^{\Gamma_m} \otimes \tau_d^{\otimes (n-m)} \right\|_1, \\
        &= \sum_{m=2}^n \frac{1}{2} \left\|\DC_m(\sigma^{(m)})^{\Gamma_m} \right\|_1, \\
        &\leq \sum_{m=2}^n \frac{m-1}{d+m-1}, \quad &\text{Lemma}~\ref{lem:last-leg-depolarize-trace-norm-bound},\\
        &\leq \sum_{m=2}^n \frac{m-1}{d}, 
    \end{align}
    which upon summing yields the desired bound
    \begin{align}
        \abs{\tr{M \Delta}} \leq \frac{n(n-1)}{2 d} = O\left(\frac{n^2}{d}\right).
    \end{align}
\end{proof}

\begin{lemma}[Lemma~\ref{lem:last-leg-depolarize-trace-norm-bound} Restated]
    Let $\DC_m(\sigma^{(m)})= (\sigma^{(m)} - \sigma^{(m-1)}\otimes \tau_d)$ and denote the partial transposition of the $m$-th leg as $\Gamma_m$. Then, we have \begin{align}
        \frac{1}{2}\|\DC_m(\sigma^{(m)})^{\Gamma_m}\|_1 \leq \frac{m-1}{d+m-1}.
    \end{align}
\end{lemma}
\begin{proof}
    To begin, let us rewrite the symmetric subspace projector as \begin{align}
        \Pi_m & = \frac{1}{m!}\sum_{\pi \in S_m}P_d(\pi), \\
        & = \frac{1}{m!}\left(\sum_{\pi'\in S_{m-1}}P_d(\pi') \otimes \mathbb{I}_m\right)\left(\mathbb{I} + \sum_{i=1}^{m-1}P_d((im))\right), \\
        & = \frac{1}{m}\left(\Pi_{m-1} \otimes \mathbb{I}_m\right)\left(\mathbb{I} + \sum_{i=1}^{m-1}P_d((im))\right), \\
        \implies \Pi_m^{\Gamma_m}& =  \frac{1}{m}\left(\Pi_{m-1} \otimes \mathbb{I}_m\right)\left(\mathbb{I} + \sum_{i=1}^{m-1}d \psi_{i,m}\right),
    \end{align}
    where $\psi_{i,m}$ is the projector onto the maximally entangled state between system $i$ and $m$. To ease the remaining manipulation, define \begin{align}
        A_m \coloneq \Pi_{m-1} \otimes \mathbb{I}_m \quad \text{and} \quad \Psi_m\coloneq \sum_{i=1}^{m-1}\psi_{i,m},
    \end{align}
    which allows us to rewrite \begin{align}
        \Pi_m^{\Gamma_m} = \frac{1}{m}A_m(\mathbb{I}+d \Psi_m).
    \end{align}
    Next, observe that $[P_d(\pi'),\Psi_m] =0$. Indeed, because $P_d(\pi')\psi_{i,m}P_d(\pi')^\dagger = \psi_{\pi'(i),m}$, we may calculate \begin{align}
        P_d(\pi')\Psi_m P_d(\pi')^\dagger = \sum_{i=1}^{m-1}\psi_{\pi'(i),m} = \sum_{j=1}^{m-1}\psi_{j,m} = \Psi_m,
    \end{align}
    and it follows from linearity that $[A_m,\Psi_m] = 0$. Thus $0 \preceq A_m\Psi_m$, which in turn implies $\|A_m\Psi_m\|_1 = \tr{A_m\Psi_m}$. With this in mind, we can strategically take partial traces to obtain \begin{align}
        \|A_m\|_1 = d D_{m-1} \quad \text{and} \quad \|A_m\Psi_m\|_1 = (m-1)\cdot \frac{D_{m-1}}{d}.
    \end{align}
    With all of these facts, we may write the partially transposed single-erasure operator as \begin{align}
        \DC_m(\sigma^{(m)})^{\Gamma_m} = \frac{1}{D_m}\Pi_m^{\Gamma_m} - \frac{1}{dD_{m-1}}\Pi_{m-1}\otimes \mathbb{I} = \left(\frac{1}{mD_m}-\frac{1}{d D_{m-1}}\right)A_m + \frac{d}{m D_m}A_m\Psi_m
    \end{align}
    A final application of the triangle inequality yields \begin{align}
        \|\DC_m(\sigma^{(m)})^{\Gamma_m}\|_1 \leq \abs{\frac{1}{mD_m}-\frac{1}{d D_{m-1}}}dD_{m-1}+\frac{d}{mD_m}(m-1)\cdot \frac{D_{m-1}}{d} = 2\frac{m-1}{d+m-1}.
    \end{align}
    Dividing by two yields the desired result.
\end{proof}

In Ref.~\cite{gong2024sample}, the authors prove an $\varepsilon$-dependent lower bound on purity estimation using the tree formalism. We note that our proof of the above theorem can easily be adapted to recover this $\varepsilon$-dependent bound as well.

\begin{proposition}[Proposition~\ref{prop:eps-dep-purity-LB} Restated]
For $0<\eps\leq1-1/d$, any PPT-BOTH purity tester with distance parameter $\eps$ requires
$n=\Omega(\sqrt{d/\eps})$
samples.
\end{proposition}

To match the testing parameter exactly, we distinguish the mixing parameter $\delta$ from the promised distance $\eps$. For $\rho_\delta
    \coloneq
    (1-\delta)\ketbra{0}{0}+\delta\tau_d,$
the closest pure state is $\ketbra{0}{0}$, and hence
\begin{align}
    \min_{\psi\ {\rm pure}}
    d_{\rm tr}(\rho_\delta,\psi)
    =
    \delta\left(1-\frac{1}{d}\right).
\end{align}
We therefore set
\begin{align}
    \delta
    \coloneq
    \frac{\eps}{1-1/d}
    =
    \frac{d}{d-1}\eps.
\end{align}
Since $\eps\leq\delta\leq2\eps$ for $d\geq2$, this rescaling does not affect the asymptotic bound.

\begin{proof}
    An $\eps$-purity tester can distinguish the Haar-random pure-state ensemble from the Haar orbit of $\rho_\delta$. Their $n$-copy average states differ by
    \begin{align}
        \Delta_\delta
        &\coloneq
        \sigma^{(n)}
        -
        \underset{U\sim{\rm Haar}}{\mathbb{E}}
        \left[
            (U\rho_\delta U^\dagger)^{\otimes n}
        \right] \\
        &=
        \sum_{S\subseteq[n]}
        (1-\delta)^{n-s}\delta^s
        \left(
            \sigma^{(n)}
            -
            \sigma_{S^c}^{(n-s)}
            \otimes\tau_{d,S}
        \right),
        \qquad s\coloneq|S|.
    \end{align}
    For each $S$, we may relabel the registers so that the $s$ maximally mixed factors occupy the last $s$ positions. Such relabeling preserves the PPT-BOTH constraint, and the same telescoping argument gives
    \begin{align}
        \abs{\tr{M\Delta_\delta}}
        &\leq
        \frac{1}{2}
        \sum_{s=0}^{n}
        {n\choose s}
        (1-\delta)^{n-s}\delta^s
        \sum_{m=n-s+1}^{n}
        \left\|
            \DC_m(\sigma^{(m)})^{\Gamma_m}
        \right\|_1 \\
        &\leq
        \frac{1}{d}
        \sum_{s=0}^{n}
        {n\choose s}
        (1-\delta)^{n-s}\delta^s
        \left(
            ns-\frac{s(s+1)}{2}
        \right) \\
        &\leq
        \frac{n^2\delta}{d}
        =
        O\left(\frac{n^2\eps}{d}\right),
    \end{align}
    where we used Proposition~\ref{lem:last-leg-depolarize-trace-norm-bound} and
    \begin{align}
        \sum_{s=0}^{n}
        {n\choose s}
        (1-\delta)^{n-s}\delta^s s
        =
        n\delta.
    \end{align}
    Constant distinguishing bias therefore requires
    $n=\Omega(\sqrt{d/\eps})$.
\end{proof}

\section{Product Testing Details}
   This appendix follows closely along the lines of Ref.~\cite{beckey2025Product}. Let us assume that a PPT-BOTH product tester exists. That is, there exists $\{M,I-M\}$ with $0\leq M \leq I$ and $0 \leq M^{\Gamma_S} \leq I$ for all $S \subseteq [n]$ for which the following hold: 
    \begin{itemize}
        \item If $\ket{\psi}$ is product, then $\tr{M \psi^{\otimes n}} \geq \frac{2}{3} =: c $
        \item If $\ket{\psi}$ is $\epsilon$-far from product, the $\tr{M \psi^{\otimes n}} \leq \frac{1}{3} =: s$.
    \end{itemize}
   This implies that for any $\psi$ which is product and any $\phi$ which is far from product, we have 
    \begin{align}
        \frac{1}{3} = b \leq \tr{M\left(\psi^{\otimes n}-\phi^{\otimes n}\right)}.
    \end{align}
    Since this holds for all states in those respective sets, and because the trace is linear, we can take an expectation over an ensemble of states that are product and an ensemble of state which are all far from product. We have
    \begin{align}
        \frac{1}{3} = b \leq \abs{\tr{M \left(\underset{\ket{\psi_A},\ket{\psi_B} \in \mathbb{C}^d}{\mathbb{E}} [\psi_A^{\otimes n} \otimes \psi_B^{\otimes n}] - \underset{\ket{\psi} \in \mathbb{C}^{d^2}}{\mathbb{E}} [\psi^{\otimes n}]\right)}}.
    \end{align}
    Now, it is clear that the ensemble on the left contains only product states. However, the ensemble on the right clearly cannot only contain states that are far from product. Among all Haar random states, either $\ket{\psi}$ is close to product, or it is far, thus we have
    \begin{align}
        \underset{\ket{\psi} \in \mathbb{C}^{d^2}}{\mathbb{E}} \left[\tr{M \psi^{\otimes n}}\right] 
        &= \underset{\ket{\psi} \in \mathbb{C}^{d^2}}{\mathbb{E}} [\tr{M \psi^{\otimes n}} \rvert \psi~\mathrm{far}] \cdot \mathrm{Prob}\left[ \psi~\mathrm{far} \right] +  \underset{\ket{\psi} \in \mathbb{C}^{d^2}}{\mathbb{E}} [\tr{M \psi^{\otimes n}} \rvert \psi~\mathrm{close}] \cdot \mathrm{Prob}\left[ \psi~\mathrm{close} \right],\\
        &\leq \frac{1}{3} \cdot \mathrm{Prob}\left[ \psi~\mathrm{far} \right] +  1 \cdot \mathrm{Prob}\left[ \psi~\mathrm{close} \right],\\
        &= \frac{1}{3} \cdot (1-\mathrm{Prob}\left[ \psi~\mathrm{close} \right]) +  1 \cdot \mathrm{Prob}\left[ \psi~\mathrm{close} \right],
    \end{align}
    which implies 
    \begin{align}
        \underset{\ket{\psi} \in \mathbb{C}^{d^2}}{\mathbb{E}} \left[\tr{M \psi^{\otimes n}}\right]&\leq \frac{1}{3} + \frac{2}{3} \cdot \mathrm{Prob}\left[ \psi~\mathrm{close} \right]
    \end{align}
    where we have bounded the first expectation by 1/3 from the assumption of the tester and the other trivially because probabilities are always at most 1. In this form, it is clear that we need to prove that the probability of obtaining a state that is close to a product state when sampling uniformly over $\mathbb{C}^{d^2}$ is highly unlikely. 
    
    \begin{align}
        \mathrm{Prob}\left[ \psi~\mathrm{close} \right] &:= \mathrm{Prob}\left[ \exists \ \text{product } \phi \ \text{s.t. } d_{\rm tr}\big(|\psi\rangle,|\phi\rangle\big) \le \varepsilon \right],\\
        &= \mathrm{Prob}\left[ \underset{\ket{\phi} \mathrm{product}}{\max} |\braket{\phi}{\psi}|^2 \geq 1-\varepsilon^2 \right], \\
        &=\mathrm{Prob}[\lambda_{\rm max} (\psi_A)  \geq 1-\varepsilon^2],\\
        &\leq Ce^{-cd},
    \end{align}
    where we have defined $\lambda_{\rm max} (\psi_A) := \|\gamma_A \|_{\infty}$ to be the maximal Schmidt coefficient and used concentration results from Ref.~\cite{hayden2006Aspects}.We can then utilize the following Corollary from Ref.~\cite{harrow2013Church}. Thus, using these ensembles only alters the bias by some $o(1)$ factor. Since the tester accepts product states with probability at least $2/3$,
\begin{align}
\tr{M(\rho_{\mathrm{prod}}^{(n)}-\rho_{\mathrm{Haar}}^{(n)})}
\geq
\frac13-\frac23
\Pr_{\psi\sim\mathrm{Haar}}
\!\left[
d_{\rm tr}(\psi,\mathrm{Prod})<\eps
\right]
=
\frac13-o(1).
\end{align} We are now ready to prove Theorem~\ref{thm:product-testing-LB} from the main text.

\begin{proof}[Proof of Theorem~\ref{thm:product-testing-LB}]
    As established in the preceding section, any PPT-BOTH measurement $\{M, \mathbb{I}-M\}$ functioning as a valid $\varepsilon$-product tester must have $\Omega(1)$ bias in the task of distinguishing 
    \begin{align}
        \Delta \coloneq \underset{\ket{\psi_A}, \ket{\psi_B} \in \mathbb{C}^{d}}{\mathbb{E}}\left[\psi_A^{\otimes n} \tilde{\otimes} \psi_B^{\otimes n}\right] - \underset{\ket{\psi}\in \mathbb{C}^{d^2}}{\mathbb{E}}\left[\psi^{\otimes n}\right],
    \end{align}
    where we denote $\rho_{\mathrm{prod}}^{(n)}$ and $\rho_{\mathrm{far}}^{(n)}$, respectively. If we can show such a measurement also has $O(n^2/d)$ bias, then it would follow that the sample complexity $n = \Omega(\sqrt{d})$. First, we can rewrite the bias \begin{align}
        \abs{\tr{M\Delta}} = \abs{\tr{M(\rho_{\mathrm{prod}}^{(n)} - \tau_{d^2}^{\otimes n} - (\rho_{\mathrm{far}}^{(n)} - \tau_{d^2}^{\otimes n}))}} \leq \abs{\tr{M(\rho_{\mathrm{prod}}^{(n)}  - \tau_{d^2}^{\otimes n})}} + \abs{\tr{M(\rho_{\mathrm{far}}^{(n)} - \tau_{d^2}^{\otimes n})}}.
    \end{align}
        From Theorem~\ref{thm:eps-ind-PPT-purity-LB}, by replacing $d \to d^2$, we see the second term is $O(n^2/d^2)$. Thus it suffices to show the first term is $O(n^2/d)$. In this setting, we define \begin{align}
       \rho_{\mathrm{prod}}^{(m)} \coloneq \sigma_A^{(m)} \tilde{\otimes} \sigma_B^{(m)},
       \end{align}
        where $\sigma_A^{(m)}$ and $\sigma_B^{(m)}$ are the maximally mixed states on the symmetric subspaces for Alice and Bob, respectively. Note that $\rho_{\mathrm{prod}}^{(1)} = \sigma_A^{(1)} \tilde{\otimes} \sigma_B^{(1)} = \tau_{d^2}$. Using a telescoping sum, we can then rewrite \begin{align}
        \rho_{\mathrm{prod}}^{(n)} - \tau_{d^2}^{\otimes n} & = \sum_{m=2}^n \left[\rho_{\mathrm{prod}}^{(m)} \otimes \tau_{d^2}^{\otimes(n-m)} - \rho_{\mathrm{prod}}^{(m-1)} \otimes \tau_{d^2}^{\otimes(n-m+1)}\right], \\
        & = \sum_{m=2}^n \left[(\sigma_A^{(m)}\tilde{\otimes} \sigma_B^{(m)}) - (\sigma_A^{(m-1)}\otimes \tau_d)\tilde{\otimes} (\sigma_B^{(m-1)}\otimes \tau_d) \otimes \tau_{d^2}^{\otimes (n-m)}\right], \\
        & = \sum_{m=2}^n \left[(\sigma_A^{(m)} - \sigma_A^{(m-1)}\otimes \tau_d)\tilde{\otimes}\sigma_B^{(m)} + (\sigma_A^{(m-1)}\otimes \tau_d)\tilde{\otimes} (\sigma_B^{(m)} - \sigma_B^{(m-1)}\otimes \tau_d) \right] \otimes \tau_{d^2}^{\otimes (n-m)},   \\
        & = \sum_{m=2}^n\left[\DC_m(\sigma_A^{(m)}) \tilde{\otimes} \sigma_B^{(m)} + (\sigma_A^{(m-1)}\otimes \tau_d) \tilde{\otimes} \DC_m(\sigma_B^{(m)})\right] \otimes \tau_{d^2}^{\otimes (n-m)}.
    \end{align}
    With Lemma~\ref{lem:last-leg-depolarize-trace-norm-bound}, the bias can then be upper bounded as \begin{align}
        \abs{\tr{M(\rho_{\mathrm{prod}}^{(n)}-\tau_{d^2}^{\otimes n})}} & = \abs{\tr{M\left(\sum_{m=2}^n\left[\DC_m(\sigma_A^{(m)}) \tilde{\otimes} \sigma_B^{(m)} + (\sigma_A^{(m-1)}\otimes \tau_d) \tilde{\otimes} \DC_m(\sigma_B^{(m)})\right] \otimes \tau_{d^2}^{\otimes (n-m)}\right)}},   \\
        & \leq \sum_{m=2}^n \abs{\tr{M\left(\left[\DC_m(\sigma_A^{(m)}) \tilde{\otimes} \sigma_B^{(m)} + (\sigma_A^{(m-1)}\otimes \tau_d) \tilde{\otimes} \DC_m(\sigma_B^{(m)})\right] \otimes \tau_{d^2}^{\otimes (n-m)}\right)}},   \\
        & = \sum_{m=2}^n \abs{\tr{M^{\Gamma_m}\left(\left[\DC_m(\sigma_A^{(m)}) \tilde{\otimes} \sigma_B^{(m)} + (\sigma_A^{(m-1)}\otimes \tau_d) \tilde{\otimes} \DC_m(\sigma_B^{(m)})\right] \otimes \tau_{d^2}^{\otimes (n-m)}\right)^{\Gamma_m}}}, \\
        & \leq \frac{1}{2}\sum_{m=2}^n \left\|\left[\DC_m(\sigma_A^{(m)}) \tilde{\otimes} \sigma_B^{(m)} + (\sigma_A^{(m-1)}\otimes \tau_d) \tilde{\otimes} \DC_m(\sigma_B^{(m)})\right]^{\Gamma_m}\right\|_1 , \\   
        & \leq \frac{1}{2}\sum_{m=2}^n \left\|\DC_m(\sigma_A^{(m)})^{\Gamma_m}\right\|_1\cdot \left\|(\sigma_B^{(m)})^{\Gamma_m}\right\|_1 + \left\|(\sigma_A^{(m-1)}\otimes \tau_d)\right\|_1 \cdot \left\|\DC_m(\sigma_B^{(m)})^{\Gamma_m}\right\|_1, \\
        & \leq \frac{1}{2}\sum_{m=2}^n \frac{m-1}{d}\cdot \left\|(\sigma_B^{(m)})^{\Gamma_m}\right\|_1 + \left\|(\sigma_A^{(m-1)}\otimes \tau_d)\right\|_1 \cdot \frac{m-1}{d}. 
    \end{align}
    It is clear that $\left\|(\sigma_A^{(m-1)}\otimes \tau_d)\right\|_1 = 1$ from it being a valid quantum state. Note that also $\left\|(\sigma_B^{(m)})^{\Gamma_m}\right\|_1 = 1$. To see this, recall from the proof of Lemma~\ref{lem:last-leg-depolarize-trace-norm-bound} that we may write $\Pi_m^{\Gamma_m} = \frac{1}{m}A_m(\mathbb{I} + d\Psi_m)$. Furthermore, $[A_m, \Psi_m] = 0$. Consequently, the entire operator $A_m(\mathbb{I} + d\Psi_m) = A_m + d A_m\Psi_m \succeq 0$. Thus $(\sigma_B^{(m)})^{\Gamma_m} = \frac{1}{D_m}\Pi_m^{\Gamma_m} \succeq 0$. Because trace is invariant under partial transposition, $\left\|(\sigma_B^{(m)})^{\Gamma_m}\right\|_1 = \tr{\sigma_B^{(m)}} = 1$.
    This allows us to complete the bound of the bias, as we can now write \begin{align}
        \abs{\tr{M(\rho_{\mathrm{prod}}^{(n)}-\tau_{d^2}^{\otimes n})}} & \leq 2 \sum_{m=2}^n \frac{m-1}{d} = O\left(\frac{n^2}{d}\right).
    \end{align}
\end{proof}

\end{document}